\documentclass[letterpaper, 10 pt, conference]{ieeeconf}  

\IEEEoverridecommandlockouts                              

\usepackage{amsfonts,amsmath, amssymb, mathtools, xcolor, hyperref, graphicx,algorithm, bbm, bm, algpseudocode, booktabs, float, optidef,subcaption,cite}

\usepackage{balance}
\usepackage{amsthm}
\usepackage{placeins}

\theoremstyle{plain} 
\newtheorem{theorem}{Theorem}
\newtheorem{lemma}{Lemma}
\newtheorem{proposition}{Proposition}

\theoremstyle{definition} 

\theoremstyle{remark}
\newtheorem{remark}[theorem]{Remark}
\newtheorem*{remark*}{Remark} 
\title{\LARGE \bf On
Online Control of Opinion Dynamics
}

\author{Sheryl Paul$^{1}$, Leslie Cruz Juarez$^{1}$, Jyotirmoy V. Deshmukh$^{1}$, Ketan Savla$^{1}$
\thanks{$^{1}$ University of Southern California,
Los Angeles, California, 90089, USA. {\tt\small \{sherylpa, lcruzjua, jdeshmuk, ksavla\}@usc.edu}}%
}

\begin{document}

\maketitle
\thispagestyle{empty}
\pagestyle{empty}

\begin{abstract}
Networked multi-agent dynamical systems have been used to model
how individual opinions evolve over time due to the opinions of other
agents in the network. Particularly, such a model has been used to 
study how a planning agent can be used to steer opinions in a desired
direction through repeated, budgeted interventions. In this paper,
we consider the problem where individuals' susceptibilities to external 
influences are unknown. We propose an online algorithm that alternates 
between estimating this susceptibility parameter, and using the current 
estimate to drive the opinion to a desired target. We provide conditions 
that guarantee stability and convergence to the desired target opinion
when the planning agent faces budgetary or temporal constraints. Our
analysis shows that the key advantage of estimating the susceptibility 
parameter is that it helps achieve near-optimal convergence to the target
opinion given a finite amount of intervention rounds, and, for a given intervention budget, quantifies how close the opinion can get to the desired target.

\end{abstract}

\section{INTRODUCTION}
In this paper, we study the problem of influencing the preferences of {\em
social agents} through the lens of control design. A social agent is a unit of a
networked multi-agent system, where the opinion or preference of each agent in
the network evolves over time while being influenced by the opinion of the
neighboring agents in the network. Additionally, we assume a designated {\em
planner agent} that can perform interventions (or control actions), which
percolate through the network and may guide agents to a specific desired
opinion. The agent network is assumed to model the {\em susceptibility} of agents to the opinions of other agents as well as the planner agent. The objective of this paper is to study how the planner agent can
estimate the susceptibility parameters of the agent network, and how the learned parameters can be
used to design resource allocation to move
opinions towards a desired value.

There are many real-world applications of such a problem setting. For example, 
consider social settings where agents are individual humans in
a social network, and a central planner may wish to encourage certain kinds of
behaviors such as sustainable consumer behavior, public health awareness, civic
engagement, etc. In these scenarios, influences can be applied through
interventions like advertising, discounts, or targeted messaging. However, these
deployments also come with a cost (e.g. time spent or monetary expense). Thus, we wish to
optimize the cost of control actions while being able to reach the desired
objective in the least amount of time.

This problem of the allocation of influencing interventions in opinion dynamics has been previously studied
with various models. Classical frameworks like DeGroot explain consensus formation under fixed interaction rules \cite{degroot}. Subsequent work (e.g., \cite{acemoglu2011opinion,peralta2022opinion}) studies the implications of Bayesian and non-Bayesian social learning models, highlighting when beliefs converge to truth, and when network structure or agent behavior leads to polarization, or clustering.
The Friedkin–Johnsen model
introduces agent stubbornness \cite{Friedkin_Johnsen_2011}, and has been
extended to study leadership design (tailored to network structure)
\cite{wang_adding_2025}, adversarial settings involving competing stubborn
agents \cite{shrinate_leveraging_2025}, and convergence behavior over
time-varying or concatenated graphs
\cite{proskurnikov_opinion_2017,wang_achieving_2021,lee_unified_2024}.
Optimal control formulations target specific opinion outcomes within budget
or effort constraints \cite{kozitsin2022optimal}, though often under the
assumption that the interaction parameters are known. Recent works emphasize parameter recovery and forecasting e.g., sublinear or Expectation Maximization-style estimators \cite{neumann2024sublinear,monti2020learning} and sociologically informed neural predictors \cite{okawa2022predicting}.

Such Bayesian estimation methods use regret minimization and bandit algorithms that update parameter estimates dynamically upon the arrival of new data. In contrast to the regret-minimization focus of online-learning and bandit methods \cite{abbasi2011improved,bastani2020online,cesa2006prediction} and regret-optimal estimation and control methods \cite{goel2023regret}, we treat opinion shaping as a networked control problem.
Unlike~\cite{okawa2022predicting}, which forecasts opinions with sociologically-informed neural nets, we aim to control an opinion dynamic model to a prescribed target under budget and constraints. 
We show that by alternating parameter identification with a control law (derived analytically), we can drive the system to a prescribed target opinion, while providing explicit convergence-rate and cost bounds.


Our proposed method combines opinion dynamics inspired by the Friedkin-Johnsen model and online control, where we assume that the `susceptibility' parameters for the network are unknown and are estimated. The algorithm alternates between exploration (to maximize parameter accuracy) and exploitation (using an optimal control law to drive opinions towards a desired target under current estimates). In contrast to~\cite{kth}, our control law is derived analytically, and we provide proofs of convergence for state and parameter errors over our two-phase algorithm. We also address a  {\em feasibility} problem: given a time horizon, and budget-constraint, determining if the system can be steered to be within a user-specified error tolerance from the desired target opinion. We summarize our contributions as follows:

\textbf{Contributions. }
(1) We analyze the state invariance, equilibria and convergence of opinion dynamics model, derive analytic control update rules, and prove convergence to the target when the parameters are under defined constraints. We further formulate a finite-horizon feasibility problem under error (distance of the state to the target) and budget constraints. (2) For unknown susceptibility parameters, we propose an online control algorithm that alternates between parameter-estimation and control, ensuring parameter identifiability and performance. We prove that, under this algorithm, the combined parameter and state error converges, and the system state is driven to the desired target. (3) We compare our work through simulations against recent models that perform budget-constrained optimization and those that perform online control via gradient descent. Additionally, we benchmark the cost of parameter-estimation by comparison with the cost of control under known parameters.

The paper is organized as follows: Section 2 introduces the system dynamics and formulates the control and parameter-estimation problems. In Section 3 we analyze system equilibria, stability under static and time-varying inputs, and perform feasibility analysis via convergence rate and control cost. Section 4 presents the parameter identification algorithm with sufficient conditions for correctness, while Section 5 describes the adaptive control algorithm with theoretical convergence guarantees. In Section 6 we report simulation results, and discuss limitations and future directions in Section 7.

\section{PRELIMINARIES}
\noindent{\bf Friedkin-Johnsen Model for Opinion Dynamics.} The opinion dynamics model in our work is inspired by the Friedkin-Johnsen (FJ) model
\cite{fjmodel} which is shown in Equation~\eqref{eq:fjmodel}. 
\begin{equation}
\label{eq:fjmodel}
x(t)=\Lambda(I-L) x(t-1) + (I-\Lambda) \mathbf{d}  
\end{equation}
where  $x(t) \in \mathbb{R}^N$ represents the system state at time $t$,  and $x_i(t) \in \mathbb{R}$ denotes the opinion of agent $i$. The desired opinion vector is denoted by $\mathbf{d} = d\mathbf{1} \in [0,1]^N$,
with the scalar $d \in [0,1]$ representing the target opinion. \( \Lambda \) is a diagonal matrix
representing agents' susceptibility to interpersonal influence, where
$\lambda_{ii}$ represents the susceptibility of agent $i$, and $L$ is the graph-Laplacian of the underlying social network graph. This model has the unique fixed point:
$x^{*}  =  (\Lambda L + I - \Lambda)^{-1}(I-\Lambda) \mathbf{d}$. The quantity $\Lambda L$ encodes how much each agent averages neighbors’ opinions, and $I - \Lambda$ encodes attachment or bias to one's prejudice. 

We note that the FJ model does not include a planner agent that can alter the opinion dynamics. To allow for such an agent, in \cite{kth}, the authors present a modified two timescale model; in our paper, we consider only the discrete-time version of this model, which is 
shown in Equation~\eqref{eq:system_dynamics}.
\begin{equation}
x(t) =  V\big[(I-HU(t)) x(t-1) + HU(t) \mathbf d\big].
\label{eq:system_dynamics}
\end{equation}

In the above formulation, $H=\mathrm{diag}(h_1,\dots,h_N)\succ 0$ is the susceptibility matrix where $h_i$, represents the susceptibility of agent $i$ to the planner's influence. The non-negative control vector $u(t)\in\mathbb{R}^N_{\ge 0}$ represents the influence that the planner exerts on the system, with $u_i(t)$ representing the influence exerted on agent $i$ at time $t$ ($U(t) = \mathrm{diag}(u(t))$.  $V = (\Lambda L+ I - \Lambda)^{-1} (I - \Lambda)$, $V \in \mathbb{R}^{N \times N}$, is defined by the fixed point of the FJ model, where $v_{ij}$ represents the weight agent $i$ places on agent $j$’s opinion. 
In \cite{kth}, the authors argue that the matrix $-(\Lambda L+ I - \Lambda)$ is
Metzler and Hurwitz, and therefore invertible, under the assumption of a strongly connected graph and at least one non-stubborn agent i.e. $\exists i \text{ s. t. }  \lambda_{ii} <1$\footnote{The matrix  is
\textit{Metzler} because $L$ is a Laplacian  and $\Lambda$ is diagonal with
entries in $[0,1]$, preserving the Metzler property. It is \textit{Hurwitz}
because the graph $G$ is connected and has at least one non-stubborn agent ensuring it is positive definite and invertible.}.
    
The interpretation of system dynamics in \eqref{eq:system_dynamics} is that
agents form a convex blend of their prior opinions $x(t-1)$ and the target
$\mathbf{d}$, weighted by $HU(t)$. (We assume $HU(t)\in[0,1]$ (elementwise)  and will justify why this restriction is necessary in the next section.) After this, the social mixing  matrix $V$
propagates these updates across the network. \footnote{Our model applies control at every discrete time step, whereas \cite{kth} uses a two-timescale scheme in which interventions 'campaigns'+) occur at discrete epochs and the system evolves in continuous time between campaigns.}

\textit{Cost definitions and budget.} We assume that the control strategy of the planner incurs some cost. Similar to \cite{kth}, we assume
that these costs are linear functions of the control input\footnote{The cost in a social setting
models the planner purchasing influence (ads, outreach campaigns,
incentives). This, along with simpler comparison to related work acts as our motivation for choosing a linear sum of costs, over the advantages of smoother analysis provided by quadratic costs.}.
W.l.o.g., we interpret $u_i(t)$ as the actual per–agent effort or cost at time
$t$. We can define cost  at time instant $t$ as $c(t) = \sum_{i=1}^n u_i(t)$ and
the cumulative cost over horizon-$T$ as $c_u(T) = \sum_{i=1}^n \sum_{t=1}^T
u_i(t)$. A natural budget constraint is $c_u(T)\le C_{\max}$ for some $C_{\max}>0$.

\noindent{\bf Problem Statement.}
We describe the problems we wish to solve, under the cases of known and unknown parameters. Specifically, we assume the susceptibility matrix $H$ which encodes each agent’s responsiveness to planner interventions is unknown. \footnote{Susceptibility to influence is latent, and cannot typically be directly observed, or reliably inferred through survey data, as individuals rarely estimate their own susceptibility accurately. In contrast, other parameters such as the interaction matrix $V$ can be inferred from observed networks, and the opinion vector $x(t)$ is observable through expressed behavior (clicks, ratings, purchases, or survey/sentiment responses).}  
The problems we solve in this work are stated below, categorized in the known and unknown parameter settings:

\noindent{\textit{Known susceptibility parameters}}.
 \emph{(1) Global asymptotic stability:} We design a control sequence $u(t)$ so that $x=\mathbf d$ is a globally asymptotically stable equilibrium of the system, i.e., $x(t)\to\mathbf d, \text{ as } t \to \infty$.\emph{ (2) Feasibility under error constraints and budget:}
    Given a time horizon $T$, and a target accuracy $\epsilon \geq \|x(t) - \mathbf{d}\|$ and a total cost budget $C_{\max}>0$, we wish to find $u(t)$ that satisfies both the accuracy and budget constraints.\\
\noindent{\textit{Unknown susceptibility parameters}}.
\emph{(3) Online Control:} We employ a parameter estimation algorithm, to estimate the unknown parameter $H$, while simultaneously driving $x$ to the target using adaptive control. We show that, under certain conditions the combined error converges: the state error decreases to zero, while the parameter error remains bounded at a minimum level.

\section{ANALYSIS OF SYSTEM DYNAMICS}
As a starting point, we examine the boundedness of the state trajectory and establish that under the system dynamics, $x(t)$ remains confined to an invariant subset of $\mathbb{R}^n$ for all time.
\begin{theorem}
    If $V \in \mathbb{R}^{n \times n}$ is row-stochastic matrix, i.e., $V\mathbf{1}=\mathbf{1}$ and $v_{ij} \geq 0, \ \forall i,j$, and the control vector $u(t)$ satisfies $0\le h_i u_i(t) \le 1 \  \forall i$, for $\forall t$, assuming that the opinion vector $x(t-1) \text{ lies in the subset } [0,1]^n \text{ with the desired target } \mathbf{d} \in [0,1]^n$, then under the dynamics equation \eqref{eq:system_dynamics} $x(t)$ remains in the set $[0,1]^n$, i.e., the set $[0,1]^n$ is positively invariant under the system dynamics. 
\label{thm:theorem_1}
\end{theorem}
\begin{proof}
Since each term $h_i u_i(t) \in [0,1]$, we observe that for each $i$, the quantity $(1 - h_i u_i(t)) x_i(t-1) + h_i u_i(t) d$
is a convex combination of $x_i(t-1) \in [0,1]$ and $d \in [0,1]$ and
\begin{equation}
 \therefore 0 \leq (1 - h_i u_i(t)) x_i(t-1) + h_i u_i(t) d \leq 1 \  \forall i
 \label{eq:exp_bound}\end{equation}
Therefore, the vector $[(I - H U(t)) x(t-1) + H U(t)\mathbf{d}] \in [0,1]^n$.

Now since $x(t) = V[(I - H U(t)) x(t-1) + H U(t) \mathbf{d}]$, and each row $(v_{i1}, \dots, v_{in})$ of $V$ is a probability vector (i.e. $v_{ij} \geq 0$ and $\sum_j v_{ij} = 1$), it follows that each coordinate $x_i(t)$ is a convex combination of entries of the expression $(1 - h_i u_i(t)) x_i(t-1) + h_i u_i(t) d$.

$\therefore \text{From \eqref{eq:exp_bound} it follows that:}$ 
\(
 0 \leq x_i(t) \leq 1, \ \  \forall i.
\)
 Thus, the opinion vector $x(t) \in [0,1]^n, \forall t$.
\end{proof}
\noindent We now address the problems \emph{known susceptibility parameters} setting. First, we establish \emph{global asymptotic stability of the equilibrium} (i.e., $x(t)\to \mathbf d$ as $t\to\infty$) under some bounds on the control input $u(t)$. Second, we formulate and analyze a \emph{finite-horizon feasibility} problem that certifies attainment of a prescribed error (distance from target) within a specified control budget.
\subsection{Stability Analysis under Bounded Control Input}
To establish global asymptotic stability of the equilibrium $x=\mathbf d$, we analyze the existence, uniqueness, and stability of the  equilibria under bounded control inputs.
\begin{theorem}[Stability of the desired equilibrium for admissible initial conditions]
Consider the dynamics in~\eqref{eq:system_dynamics}, if \(x(0), \mathbf{d} \in [0,1]^n\), and if the control vector $u(t)$ satisfies $0 < h_i u_i(t) < 1,\ \forall i, t$, then for all initial conditions $x(0) \in [0,1]^n$, the system converges exponentially to the equilibrium $x^* = \mathbf{d}$.
\label{thm:GAS}
\end{theorem}
\begin{proof}
    Let us define the deviation from equilibrium as $x'(t) = x(t) - \mathbf{d}$. 
    Since $\mathbf{d}$ is constant:
\(\mathbf{d} = V[(I - H U(t)) \mathbf{d} + H U(t) \mathbf{d}] = V \mathbf{d}. \)
Subtracting $\mathbf{d}$ from both sides of the system dynamics Eq.~\eqref{eq:system_dynamics}:
we get
\[
x'(t) = V (I - H U(t)) x'(t-1).
\]
Let us define $A(t) = V(I - H U(t))$. Then the update becomes:
    \[x'(t) = A(t-1) x'(t-1) = A(t-1) A(t-2) \cdots A(0) x'(0).\]

\noindent Let $\eta(t) = \min_i h_i u_i(t)$, so that $0 < \eta(t) < 1$. Then:
\begin{equation}\label{eq:at}
\begin{aligned}
\|A(t)\|_\infty &\le 1-\eta(t) < 1,\\
\|x'(t)\|_\infty &\le \Bigg(\prod_{s=0}^{t-1} (1-\eta(s))\Bigg)  \|x'(0)\|_\infty.
\end{aligned}
\end{equation}
Since $ \eta(t) \in (0,1), \ \prod_{t=0}^\infty (1 - \eta(t)) = 0$
\text{ and therefore, } $\|x'(t)\|_\infty \to 0.$
 Thus, we have $\lim_{t \to \infty} x(t) = \mathbf{d}.$
 
If $u(t) = u$ is constant, then so is $A = V(I - H U)$. By the assumption $0 < h_i u_i < 1$, we have:
\(\|I - H U\|_{\infty} = \max_i |1 - h_i u_i| < 1
\ \text{so} \ \|A\|_\infty \leq \|V\|_\infty \|I - H U\|_\infty < 1,\) and we can similarly show that $\lim_{t \to \infty} x(t) = \mathbf{d}.$\end{proof}
\begin{remark}[Joint Spectral Radius Argument.]
    We could also show this result using an argument with the joint spectral radius of $A(t)$. Define
$\mathcal{M}=\{\,V(I-HU):\ U=\mathrm{diag}(u),\ 0< h_i u_i \le 1\,\}$, 
$\eta=\inf_t \min_i h_i u_i(t)>0$. Then \(\|A\|_{\infty}\le 1-\eta\) for all \(A\in\mathcal{M}\), so \(\rho(\mathcal{M})\le 1-\eta<1\).
By standard results \cite{jungers2009jsr}, for any sequence \(A(t)\in\mathcal{M}\), $x'(t)=A(t-1)\cdots A(0)\,x'(0)\to 0 $\ \text{exponentially}, hence \(x(t)\to x^*\).
\end{remark}
    \begin{remark}[System behavior without control input.]
          If \( U = 0 \), the system reduces to the standard consensus dynamics~\cite{degroot-consensus}:
\(x(t) = V x(t-1)\). If \( x(t-1) = \mathbf{d} \), then \( x(t) = \mathbf{d} \). That is, the system remains at equilibrium. However, if $x(t-1) \neq \mathbf{d}$, then the system converges to a weighted average of the initial state components instead, and $\mathbf{d}$ does not influence the convergence.
  
\end{remark}
Having established convergence to the desired equilibrium for admissible initial conditions under bounds on the control input, we go on to characterize the convergence rate for such inputs and quantify the corresponding control cost.
\subsection{Feasibility Analysis of Control Cost and Accuracy}
\textbf{Rate of Convergence. }
We consider a time-varying control input of the form $u(t) $, which enforces a uniform rate of convergence $r(t)  = a \cdot b^t\in (0,1)$, with $0<a,b< 1$ for all agents, s.t. $h_i u_i(t) = r(t)$. Thus, the individual control input for each agent can be set as 
\begin{equation}
    u_i(t) = \frac{r(t)}{h_i} \label{eq:exp-control}
\end{equation} From \eqref{eq:at} 
\(
\|x'(t)\|_\infty \leq \|x'(0)\|_\infty \cdot \prod_{s=0}^{t-1} (1 - a b^s).
\)
Using the inequality \( 1 - y \leq e^{-y} \) for \( y \in (0,1) \), we bound the product as:
\(
 \|x'(t)\|_\infty \leq \exp\left( - \frac{a(1 - b^t)}{1 - b} \right) \cdot\|x'(0)\|_\infty\) \(\text{so } \lim_{t \to \infty} \|x'(t)\|_\infty \leq \exp\left( - \frac{a}{1 - b} \right) \cdot \|x'(0)\|_\infty
\).

And the distance from equilibrium $\epsilon$, at time $T$ can be bounded by:
\begin{equation}
\epsilon = \|x'(T)\|_\infty \leq \exp\left( - \frac{a(1 - b^T)}{1 - b} \right) \cdot \|x'(0)\|_\infty
\label{eq:xt-eq-dist-bound}
\end{equation}
 We define {\em non-uniform control} as when each agent has a specified control input depending on their individual susceptibility parameter $h_i$, and {\em uniform control} as when the same control input (say $u_c(t)$) is applied to all agents i.e., $u_i(t)$ $ = u_c(t), \  \forall i$.
\footnote{In the special case of a constant control input $u(t)=u$, convergence can be regulated by fixing a rate $r\in(0,1)$ and setting $u_i=r/h_i$, which yields $\|x'(t)\|_\infty \le (1-r)^t \|x'(0)\|_\infty$. In the case of uniform control, we impose bounds $h_{\min}\le h_i\le h_{\max}$ and select $u_c = r/h_{\max}$. This guarantees $h_i u_c \in [r h_{\min}/h_{\max},   r]\subset(0,1)$ for all $i$, leading to the contraction bound
\(
\epsilon = \|x'(t)\|_\infty  \le  \Bigl(1-\tfrac{r h_{\min}}{h_{\max}}\Bigr)^t \|x'(0)\|_\infty.
\)}
\noindent{\bf Cost of Control.}
We now compute the associated cost of control, making explicit the trade-off between convergence speed and control expenditure.
 The expression for cost can be given as:
\begin{equation}\label{eq:control-cost}
\begin{aligned}
c_u(T)
&= \sum_{t=0}^{T-1} \sum_{i=1}^n \frac{r(t)}{h_i}
 = \sum_{t=0}^{T-1} \sum_{i=1}^n \frac{a b^t}{h_i}
 = \frac{a(1 - b^T)}{1 - b}   S, \\
S
&= \begin{cases}
\displaystyle \sum_{i=1}^n \frac{1}{h_i} & \text{(non-uniform control)}\\[4pt]
\displaystyle \frac{n}{h^{\max}} & \text{(uniform control)}
\end{cases}
\end{aligned}
\end{equation}
In both cases, the cost depends on the same bilinear term $a(1-b^T)/(1-b)$ that also governs convergence speed.

\noindent{ \bf Minimum Control Cost for Finite-Time $\epsilon$-Accuracy.}
 Firstly, in order to ensure at time $t$, that the state $x(t)$ is  at most $\epsilon$ away from the equilibrium, from \eqref{eq:xt-eq-dist-bound} we have: $a (1 - b^t) \geq \left(1-b\right)  \log\left( \frac{\|x'(0)\|_\infty}{\epsilon} \right)$. Secondly, to satisfy a budget constraint, i.e., $c_u(t) \leq C_{\max}$, from equation~\eqref{eq:control-cost} we get: $a(1-b^t) \leq (1-b)\frac{C_{\max}}{S}$.
These inequalities aid us in framing and solving the problem of analyzing the feasibility of achieving $\epsilon-$accuracy under given time and budget constraints.

\noindent\textit{Feasibility Analysis.}
Given a finite horizon $T$, minimum accuracy requirement $\epsilon>0$, initial error $\|x'(0)\|_\infty$, budget limit $C_{\max}$, and the sum of susceptibility parameters $S = \sum_{i = 0}^n1/ h_i$, we seek to determine whether there exist $a,b \in (0,1)$ such that the control schedule
\(
u(t) = a b^t/h^i
\)
satisfies both the accuracy and budget constraints. This amounts to solving the constraint system:
\begin{equation}
   \frac{C_{\max}}{S}\geq a \frac{1-b^T}{1-b}  \geq  \log\!\left(\frac{\|x'(0)\|_\infty}{\epsilon}\right), 
   \  0<a,b<1.
   \label{eq:constraint-3}
\end{equation}
A solution exists if:
\begin{itemize}
\item
\textit{Condition 1.} Relating the error $\epsilon$ to the budget $C_{\max}$ we require:
\(
\epsilon  \ge  \|x'(0)\|_\infty    e^{-\frac{C_{\max}}{S}}.
\)

\item \textit{Condition 2.} We can rewrite~\ref{eq:constraint-3} as:
\[
\log\!\Big(\tfrac{\|x'(0)\|_\infty}{\epsilon}\Big)  \frac{1-b}{1-b^{T}}
\ \le\ a \ \le\
\frac{C_{\max}}{S}  \frac{1-b}{1-b^{T}} .
\]
Under Condition 1 the interval is nonempty; to also satisfy $a<1$ it suffices that $\frac{C_{\max}}{S}  \frac{1-b}{1-b^{T}} \ \le\ 1$.

\end{itemize}
Obtaining the value of $b$ requires solving a polynomial equation: $a  (1-b^{T})  -  \frac{C_{\max}}{S}  (1-b) = 0$ or $a  (1-b^{T})  -  \log\!\Big(\tfrac{\|x'(0)\|_\infty}{\epsilon}\Big)  (1-b) = 0.$
Hence, feasibility reduces to verifying these conditions, which can be checked analytically or with standard nonlinear solvers.

Thus, we have established the convergence of the system dynamics and analyzed the feasibility problem in the case of known parameters. However, if the susceptibility parameter $H$ is unknown,
estimating it enables us to design $u(t)$ accurately enough for faster convergence, while respecting the bounds $0<h_i u_i<1$ and budget constraints. Although convergence can be achieved without exact estimation of $H$, more accurate estimates allow tighter enforcement of the input bounds, thereby improving the convergence rate. This motivates the parameter identification algorithm that we introduce in the next section.

\section{PARAMETER ESTIMATION}
 In this section, we adapt standard parameter-estimation techniques for discrete-time systems following the framework presented in the literature \cite{Ioannou.Fidan:06} (Ref. Chapter 4) to identify the susceptibility parameter $H$. 
In such techniques, unknown model parameters are updated online using error signals between the plant and a reference/prediction model. An adaptation law (e.g. gradient, Lyapunov-based, or least-squares) adjusts the estimates to maintain stability and desired performance.
The canonical setup in this framework considers systems of the form:
\begin{equation}
\label{eq:adapcontrol}
x(t) = \tilde{F_t}(x(t-1), u(t-1))+ {F}_t(x(t-1), u(t-1)) \theta,
\end{equation}
where $\tilde{F_t}$ and $F_t$ are assumed to be known functions \footnote{For convenience we denote them without the arguments as $\tilde{F_t}$ and $F_t$ henceforth.}, and $\theta \in \mathbb{R}^n$ is an unknown parameter vector that appears linearly with the regressor matrix $F_t$.

We reformulate our system to match the canonical form to show how we can apply
the adaptive control technique. We denote the parameter $h_i$ as $\theta_i$ and
$H$ as $\Theta$ henceforth. We observe that the following substitutions allows us to rewrite
\eqref{eq:system_dynamics} in the form of \eqref{eq:adapcontrol}:
\[
\begin{array}{l@{\hspace{1em}}l}
\tilde{F_t}:= Vx(t-1)  &
F_t:=V \mathrm{diag}\big(u(t) \circ(\mathbf{d} - x(t-1))\big)
\end{array}
\]
where $U(t) = \mathrm{diag}(u(t))$ and $\circ$ denotes elementwise multiplication.
The parameter-estimation method prescribes computing a state prediction $\hat{x}(t)$, and updating the parameter estimate at time $t$, i.e., $\hat{\theta(t)}$ as follows:
\begin{equation}\label{eq:param-est-algo}
\begin{aligned}
\hat{x}(t) &= \tilde{F}_t + F_t \hat{\theta}(t-1),\\
\hat{\theta}(t) &= \hat{\theta}(t-1) + \Psi F_t^\top\ \big(x(t)-\hat{x}(t)\big),\  \Psi=\psi I,\ \psi>0.
\end{aligned}
\end{equation}
 where $\Psi $ is a gain matrix. Intuitively, this update corrects the estimate in the direction of the state prediction error, weighted by the regressor $F_t$. Let the parameter and state estimation errors be defined as:
\(\theta^{\text{err}}(t) := \theta - \hat{\theta}(t),\ \text{and} \ x^{\text{err}}(t) := x(t) - \hat{x}(t)\). Substituting in the above dynamics yields:
\begin{equation}
x^{\text{err}}(t) = F_t \theta^{\text{err}}(t-1) \ \ \text{and} \ \
\theta^{\text{err}}(t) = (I - \Psi F_t^\top F_t)\theta^{\text{err}}(t-1)
\label{eq:theta-err}
\end{equation}
Thus, the updates in~\eqref{eq:param-est-algo} are applied iteratively at each timestep, using the observed state to refine the parameter estimate. The mechanism works by adjusting $\hat{\theta}(t)$ in the direction that would reduce the observed state prediction error, effectively steering the parameters to better match the system dynamics over time.
We now show that the parameter error under this method converges to zero, using Lyapunov analysis.

\textbf{Parameter Error Analysis.}
To show that the parameter-error reduces geometrically, we begin by introducing a quadratic Lyapunov function:
$R(t) := \frac{1}{2} \theta^{\text{err}}(t)^\top \Psi^{-1} \theta^{\text{err}}(t).$
Using the error dynamics \eqref{eq:theta-err}, we derive conditions under which $R(t)$ decreases monotonically. The analysis shows that this requires bounds on $\|F_t\|$ and an appropriate choice of the gain parameter $\psi$. By bounding the Lyapunov difference $R(t) - R(t-1)$, we prove that the parameter error decreases at a geometric rate: $\|\theta^{\text{err}}\| \leq K.(1-\kappa)^{\frac{t}{2}}$, where $\kappa >0$ and $K>0$ are constants.

\begin{lemma}[Convergence of Parameter Estimation Error]
\label{lem:PE-contraction}
Suppose the following conditions hold:
\begin{enumerate}
    \item \textit{Boundedness:} There exists $\beta>0$ such that $\|F_t\|\le 
    \beta$ for all $t$, and the adaptation gain satisfies $\psi<2/\beta^2$.
    \footnote{Since $V$ is fixed, $x(t),\mathbf{d}\in[0,1]^n$ and 
    $0< h_i u_i(t)< 1$, it is possible to bound 
    $\|F_t\|\le \|V\| \|u(t)\circ(\mathbf{d}-x(t))\|_\infty \leq \beta$. 
    The parameter $\psi$ is a design choice and can be selected to ensure 
    $\psi < \frac{2}{\beta^2}$.}
    \item \textit{Persistent Excitation:} The regressor $F_t$ is persistently 
    exciting, i.e., there exists $\alpha>0$ such that 
    $\forall t$: $F_t^\top F_t\succeq \alpha^2 I$. 
\end{enumerate}
Then the following hold:
\begin{enumerate}
    \item For \eqref{eq:theta-err}, the function 
    $R(t):=\tfrac{1}{2}{\theta^{\mathrm{err}}(t)}^\top \Psi^{-1} 
    \theta^{\mathrm{err}}(t)$ is a discrete-time Lyapunov function.
    \item The rate of decrease of $\theta^{\text{err}}$ satisfies:
    $\|\theta^{\text{err}}(t)\| \leq \sqrt{\frac{2}{\beta^2} R(0)} 
    \left(1 - \frac{\alpha^2}{\beta^2}\right)^{t/2}$.
\end{enumerate}
\end{lemma}
\begin{proof}
By definition of $R(t)$ and using Eq.~\eqref{eq:theta-err}, we can show:
\begin{equation}\label{eq:r-lyapunov} \scalebox{1}{$
  R(t) - R(t-1) = -{x^{\text{err}}(t)}^\top M_t\, x^{\text{err}}(t);\ M_t := I - \tfrac{1}{2} F_t \Psi F_t^\top
$}
\end{equation}
\noindent Note that $\forall t$, $R(t)$ is positive. To establish that $R(t)$ is a discrete-time Lyapunov function, we show that for all $t$, $R(t) - R(t-1)$ is strictly negative by showing that $M_t \succ 0$.
Since $\Psi$ is a symmetric p.d. matrix, by the norm-eigenvalue inequality, we
have $$\lambda_{\max}(F_t\Psi F_t^\top)\le \lambda_{\max}(\Psi)\|F_t\|^2$$
Since
$\lambda_{\max}(\Psi) = \psi$, and from Condition~1, $\|F_t\|^2 < \beta^2$, we get that $\lambda_{\max}(F_t\Psi F_t^\top)\leq\psi\beta^2$. 

Applying Weyl's inequalities to bound the eigenvalues of $I - \tfrac{1}{2} F_t\Psi F_t^\top$, we 
get:
\[
\lambda_{\min}(M_t) = 1-\tfrac12 \lambda_{\max}(F_t\Psi F_t^\top)\ \ge\ 1-\tfrac12 \psi\beta^2.
\]
Also by Condition~1, $\tfrac{1}{2}\psi\beta^2 < 1$, which means
that the min eigenvalue of $M_t$ is positive, proving that
$M_t \succ 0$. Next, we establish the convergence rate of 
$\|\theta^{\text{err}}(t)\|$. Using norm-eigenvalue inequality
we have
\begin{equation}
{x^{\mathrm{err}}(t)}^\top M_t x^{\mathrm{err}}(t)
\ge \lambda_{\min}(M_t)\|x^{\mathrm{err}}(t)\|^2.
\label{eq:x-error}
\end{equation}

\noindent From Eq.~\eqref{eq:theta-err}, $x^{\text{err}}(t) = F_t 
\theta^{\text{err}}(t-1)$, and from Condition~2, 
$$F_t^\top F_t \succeq \alpha^2 I, \text{ so } \|F_t\theta^{\text{err}}(t-1)\|^2 
\ge \alpha^2\|\theta^{\text{err}}(t-1)\|^2$$
\[
\text{And } {x^{\mathrm{err}}(t)}^\top M_t x^{\mathrm{err}}(t) \ge 
\Bigl(1-\tfrac12\psi\beta^2\Bigr)\alpha^2\|\theta^{\mathrm{err}}(t-1)\|^2.
\]
Now, to bound $\theta^{\mathrm{err}}(t-1)$, we consider the definition of $R(t-1)=\tfrac12 \theta^{\mathrm{err}}(t-1)^\top \Psi^{-1}\theta^{\mathrm{err}}(t-1)$. Since $\Psi = \psi I$, $R(t-1) =\tfrac{1}{2\psi}\|\theta^{\mathrm{err}}(t-1)\|^2$, and we have $\|\theta^{\mathrm{err}}(t-1)\|^2=2\psi R(t-1)$. Substituting in Eq.~\eqref{eq:x-error}: $${x^{\mathrm{err}}(t)}^\top M_t x^{\mathrm{err}}(t)
\ \ge\ 2\psi\Bigl(1-\tfrac12 \psi\beta^2\Bigr)\alpha^2 R(t-1)$$
 Substituting in \eqref{eq:r-lyapunov} we get:
\begin{equation}
\label{eq:rt-kappa}
    \begin{aligned}
&R(t) \le (1-\kappa) R(t-1)\ \text{ and } R(t) \le (1-\kappa)^t R(0) \\
\end{aligned}
\end{equation}
$\text{ where }\kappa := 2\psi\!\left(1-\tfrac{1}{2}\psi\beta^2\right)\!\alpha^2 > 0$. Using the definition of $R(t)$, we can simplify this to:
  \begin{equation}
    \label{eq:thetaerrbound}
    \|\theta^{\text{err}}(t)\| \leq \sqrt{2 \psi R(0)} (1 - \kappa)^{t/2}.
  \end{equation}
As $\psi$ is a design parameter, we observe that setting $\psi =
\frac{1}{\beta^2}$ yields the smallest value for $(1-\kappa)$, which in turn
yields the tightest upper bound on $\|\theta^{\text{err}}(t)\|$. Substituting
$\psi = \frac{1}{\beta^2}$ in \eqref{eq:thetaerrbound} yields:
 $\|\theta^{\text{err}}(t)\| \leq \sqrt{\frac{2}{\beta^2} R(0)} \left(1 -
 \frac{\alpha^2}{\beta^2}\right)^{t/2}$.
 \end{proof}

We go on to discuss persistency of excitation (PE) (Condition 2), and the choice of $u(t)$ to maintain it.

\noindent\textbf{Persistency of Excitation Conditions.}
To enforce the PE condition $\lambda_{\min}(F_t^\top F_t)\ge \alpha^2$ 
at each time instant, we expand using $F_t=V\mathrm{diag}(y(t))$ 
with $y(t)=u(t)\circ(\mathbf{d}-x(t))$, which gives
\(\lambda_{\min}\big(F_t^\top F_t\big)
\ge \big(\min\nolimits_j |y_j(t)|\big)^2 \lambda_{\min}\big(V^\top V\big).\)
Thus, a sufficient condition to ensure PE is: 
$\min_{j}|y_j(t)|\ \ge\  \frac{\alpha}{\underline{\lambda}_V}$ where $\underline{\lambda}_V = \sqrt{\lambda_{\min}(V^\top V)}$.

 Since $y_j(t)=u_j(t)|x_j(t)-d|$, we must keep the deviation $|x_j - d|$ bounded away from zero. 
 So, we enforce a margin: $|x_j(t) - d| \geq \delta(t), \ \forall j$ where $\delta(t) \in (0,1]$, i.e., every component of the state $x(t)$ must be at least $\delta$ away from the target $\mathbf{d}$ to preserve the PE lower bound.
 We choose $\delta(t)$ that guarantees this for all $u_j(t)$ using a conservative bound on $\theta$ from its estimate and error. Define $\theta^{\max}(t)=\hat\theta(t)+\theta^{\mathrm{err}}(t)$ for a given $\alpha$ we can set:
    \begin{equation}
      \delta=\frac{\alpha}{\underline{\lambda}_V \, u^{\max}(t)} \text{ where  }u^{\max}(t)=\min_j 1/\theta^{\max}_j(t)
      \label{eq:alpha-delta-init}
    \end{equation}
 Thus PE holds by keeping $x(t)$ outside a neighborhood of radius $\delta(t)$, with $\alpha$ and $\delta(t)$ chosen to respect input bounds. We now present an online algorithm that estimates $\theta$ while steering $x(t) \to \mathbf d$.
 


\section{ONLINE CONTROL ALGORITHM}
\begin{algorithm}[t]
\caption{Online Control Algorithm}
\label{alg:two-phase}
\begin{algorithmic}[1]
\State \textbf{Input} Target $\mathbf d$, social network matrix $V$ \label{line:input}
\State \textbf{Initialize} adaptation gain $\Psi$, $\alpha_0 \text{ and }\ \delta_0$ (via \eqref{eq:alpha-delta-init}), shrink factor $c_\delta$, parameter-estimate $\hat\theta(0)$
\State Set $m\gets 0$, $t_m\gets 0$ \label{line:init}
\State \textbf{Observe} $x(0)$ \label{line:obs0}
\While{$\|x(t_m)-\mathbf d\|$ has not converged} \label{line:outer-while}
 
  \While{$|x_j(t_m)-d|>\delta_m, \ \forall j$ and $\alpha_m> \alpha_{\min}$} \label{line:explore-while}
    \State $u_j(t_m)\gets \frac{\alpha_m}{\underline{\lambda}_V|x_j(t_m)-d_j|}, \ \forall j$ \label{line:explore-u}
    \State \textbf{Predict} $\hat x(t_m + 1)$; \textbf{observe} $x(t_m + 1)$ \label{line:pred-obs-explore}
    \State \textbf{Parameter update} $\hat\theta(t_m + 1)$ via~\eqref{eq:param-est-algo}\label{line:param-update}
    \State $t_m\gets t_m + 1$ \label{line:t-increment-explore}
  \EndWhile
   \State Update parameter bound: $\theta^{\max}\gets \hat\theta(t_m)+\theta^{\mathrm{err}}(t_m)$ \label{line:theta-max}
  \While{$|x_j(t_m) - d| \geq \gamma \delta_m, \ \forall j$ } \label{line:exploit-for}

    \State Set contraction rate: $r(t_m) = ab^{t_m}$
    \State $u_j(t_m)\gets \tfrac{r(t_m)}{\theta^{\max}_j},\ \forall j$ \label{line:exploit-u}
    \State \textbf{Observe} $x(t_m+1)$;\quad $t_m\gets t_m+1$ \label{line:obs-exploit}
  \EndWhile
  \State \textbf{Update schedule} $\delta_{m+1},\ \alpha_{m+1}$ as in~\eqref{eq:alpha_delta}\label{line:update-schedule}
  \State $m\gets m + 1$ \label{line:inc-cycle}
\EndWhile
\end{algorithmic}
\end{algorithm}

\noindent Our algorithm alternates between ``exploration'' and ``exploitation'' phases indexed by $m$:  Let the number of steps in the exploration and exploitation phases of cycle $m$ be $K^m_\theta$, and $K_x^m$ resp. Each cycle $m$ is additionally characterized by a neighborhood radius $\delta_m$ and a PE level $\alpha_m$:

\noindent (i) \emph{Exploration:} executed while $|x_j(t_m) - d| > \delta_m, \ \forall j$  to guarantee PE and drive $\hat\theta(t_m)\to\theta$. We develop a controller that ensures PE:
    \begin{equation}
        u_j(t) = \operatorname{clip}_{[0,\,1/\theta_j]} \left(\frac{\alpha}{\underline{\lambda}_V  |x_j(t)-d|}\right), \ \forall j\label{eq:PE-control}
    \end{equation}
\noindent (ii) \emph{Exploitation:} executed once $ \exists j, \text{ s.t. } |x_j(t_m) - d| 
\leq\delta_m$ using the analytic control input (Eq.~\eqref{eq:exp-control}).
To prevent over-contraction in this phase and preserve a nontrivial $\delta_{m+1}$ for the next cycle's exploration phase, we impose: $|x_j(t_m)-\mathbf{d}| \ge \gamma \delta_m, \ \forall j $ with $\gamma\in(0,1)$.

 \noindent For a chosen $\alpha_0$, we initialize $\delta_0$ using Eq.~\eqref{eq:alpha-delta-init}. At the end of each cycle $m$ (after the exploitation phase), to guarantee re-entry into the exploration phase in the next cycle, we shrink $\delta_m$ by some factor and set:
\begin{equation}
\delta_{m+1}=c_\delta (\gamma \delta_m) \text{ and } \alpha_{m+1}=\underline{\lambda}_V  u^{\max} (t^{\mathrm{end}}_{m}) \delta_{m+1}
\label{eq:alpha_delta}
\end{equation}
where $c_\delta\in(0,\tfrac12]$ is a shrinking factor, and $t^{\mathrm{end}}_{m}$ denotes the time at the end of cycle $m$. This rule ensures $|x_j(t_{m+1})-\mathbf d|> \delta_{m+1}, \ \forall j$, so the next cycle begins in exploration with a PE level that remains feasible under the input bounds.

\noindent We carry out the exploration phase until $\alpha_m>\alpha_{\min}$ where $\alpha_{\min} \in (0,1)$. Keeping $\alpha>\alpha_{\min}$ forces $|x_j(t_m) - d| > \delta_m, \forall j$ hindering the state convergence to the target i.e. $\|x(t)-\mathbf d\| \to 0$. Upon reaching $\alpha_{\min}$, we stop the exploration phases and define the parameter error at this point as $\theta^{\mathrm{err}}_{\min}$. Let
 $m_\star=\min\{m:\alpha_m\le\alpha_{\min}\}$, define 
\begin{equation}
\theta^{\mathrm{err}}_{\min}=\theta^{\mathrm{err}}\bigl(t^{\mathrm{end}}_{m_\star}\bigr),\  
R_{\min}=\tfrac12\,\theta^{\mathrm{err}}_{\min}{}^\top \Psi^{-1}\theta^{\mathrm{err}}_{\min} \label{eq:theta-min-r-min}
\end{equation}
We continue exploitation letting $\delta_m\to0$ to drive $x(t)\to\mathbf d$.

We can now present our online control algorithm~\ref{alg:two-phase}.
and describe it as follows: Lines \ref{line:input}–\ref{line:obs0} specify the problem inputs and perform initialization: the current state is observed, the parameter-estimate and other hyperparameters are set. Lines \ref{line:explore-while}–\ref{line:t-increment-explore} execute the exploration phase while $|x_j(t_m)- d|_\infty>\delta_m, \ \forall j$ and $\alpha_m >\alpha_{\min}$. The control input ensuring the $PE$ condition is applied in line \ref{line:explore-u} and the parameter-estimate is updated via line \ref{line:param-update}. Lines \ref{line:exploit-for}–\ref{line:obs-exploit} perform exploitation within the bounds $\gamma\delta_m \le |x_j(t_m)- d| \le \delta_m, \ \forall j$  using the analytic control with the conservative bound $\theta^{\max}$ (lines \ref{line:theta-max}–\ref{line:exploit-u}).
At the end of the cycle, the neighborhood radius $\delta_m$ and PE level $\alpha_m$ are updated in line \ref{line:update-schedule} using Eq. \eqref{eq:alpha_delta}, and the cycle counter advances (line \ref{line:inc-cycle}).

Once $\alpha_m \leq \alpha_{\min}$, the exploration phase in each cycle, is no longer carried out. However, the exploitation phase (lines~\ref{line:exploit-for}-\ref{line:obs-exploit}), keeps repeating and as $\delta_m \to 0$, the state error $\|x(t_m)- \mathbf{d}\| \to 0$, as our control input (line \ref{line:exploit-u}) satisfies the conditions stated in Theorem ~\ref{thm:GAS}.

\textbf{Combined State and Parameter errors:}
Let $t_m$ be the start and end of cycle $m$. Define the combined error:
\begin{equation}
   \mathcal{E}_m:=\nu_\theta R(t_m)+\nu_x\|x(t_m)-\mathbf{d}\|^2
   \label{eq:combined-error-def}
\end{equation}
with $\nu_\theta>0$ weighting the parameter error defined by the Lyapunov function $R(t_m)$, and $\nu_x>0$ weighting the distance of the state to the target $\|x(t_m)-\mathbf{d}\|^2$. Now we go on to show the convergence of this combined error under Algorithm~\ref{alg:two-phase}.

\begin{proposition}[Convergence of combined error]\label{prop:delta-cycle-shrink}
Let $t_m$ and $t_{m+1}$ be the start of cycles $m$ and $m+1$ respectively. Under the definitions of combined error~\eqref{eq:combined-error-def} and minimum parameter error~\eqref{eq:theta-min-r-min},  
the state converges to the target: $\|x(t_m)-\mathbf d\|\to 0$ as $m\to\infty$ while the parameter error reaches a minimum $R(t_m)\to R_{\min}$; hence 
\(
\mathcal{E}_m \to \nu_\theta R_{\min}
\)
\end{proposition}
\begin{proof}
    During exploration ($\|x(t_m)-\mathbf d\|_\infty>\delta_m$), our control input as designed in Eq.~\eqref{eq:PE-control}, fulfills the conditions in Lemma~\ref{lem:PE-contraction}. Hence, the per step contraction is:
\(R(t_m{+}1)\ \le\ (1-\kappa_m) R(t_m)\) (from Eq.~\eqref{eq:rt-kappa}). After $K_\theta^{(m)}$ exploration steps and noting that in exploitation $\hat\theta$ is not updated (so $R$ does not increase),
\begin{equation}\label{eq:param-bound-proof}
R(t_{m+1})\ \le\ (1-\kappa_m)^{K_\theta^{(m)}} R(t_m).
\end{equation}
The state update satisfies the contraction condition (as in Theorem ~\ref{thm:GAS}), which implies
$\|x(t_m{+}1)-\mathbf d\|_\infty\le \|x(t_m)-\mathbf d\|_\infty$. Hence, once the trajectory enters the $\delta_m$ neighborhood, 
$\|x(t_{m+1})-\mathbf d\|_\infty\le \delta_m$, and using the transformation of $\ell_\infty$ norm to $\ell_2$ norm:
\begin{equation}\label{eq:state-bound-proof}
\|x(t_{m+1})-x^*\|^2\ \le\ n \delta_m^2.
\end{equation}
Multiplying \eqref{eq:param-bound-proof} by $\nu_\theta$ and \eqref{eq:state-bound-proof} by $\nu_x$ and adding gives
\[
\mathcal{E}_{m+1}\ \le\ \nu_\theta(1-\kappa_m)^{K_\theta^{(m)}}R(t_m)\ +\ \nu_x n \delta_m^2.
\]
Using $\delta_m=c_\delta\|x(t_m)-\mathbf d\|_\infty\le c_\delta\|x(t_m)-\mathbf d\|$ yields:
\begin{equation}
    \begin{aligned}
        \mathcal{E}_{m+1}
\ &\le\ (1-\kappa_m)^{K_\theta^{(m)}} \nu_\theta R(t_m)\ +\ n c_\delta^2 \nu_x \|x(t_m)-\mathbf{d}\|^2 
    \end{aligned}
\end{equation}
Finally, the updates on
$\alpha_{m+1}$ and $\delta_{m+1}$ using Eq.~\eqref{eq:alpha_delta} guarantee feasibility of inputs and
re-entry into exploration, thereby maintaining PE in the next cycle.

Now recall the definitions of $R_{\min}$ and $\theta^{\mathrm{err}}_{\min}$ from Eq.~\eqref{eq:theta-min-r-min}. For $m\ge m_\star$, parameter updates stop, so $R(t_m)=R_{\min}$, while the state update remains contractive with $\delta_m \to 0$ implying $\|x(t_m)-\mathbf d\|\to 0$. Consequently,
\(
\mathcal{E}_m=\nu_\theta R_{\min}+\nu_x\|x(t_m)-\mathbf d\|^2 \to \nu_\theta R_{\min}\).
\end{proof}
Thus, we have shown that our algorithm sufficiently estimates the parameter $\theta$, while driving the state to the target.

\begin{remark}
    By our result in ~\ref{thm:GAS}, estimating $\theta$ is not required to drive $x(t) \to \mathbf{d}$; it suffices to choose $u(t)$ so that the condition $\theta_i u_i(t) \in (0,1), \ \forall i, t$ is met. However, as the estimation error shrinks, we can tune $u(t)$ with greater accuracy which tightens the contraction factor and improves the rate of convergence. We demonstrate this through simulations in our next section.
\end{remark}
\section{SIMULATIONS}

\begin{figure}[t]
    \centering
    \includegraphics[width=0.67\linewidth]{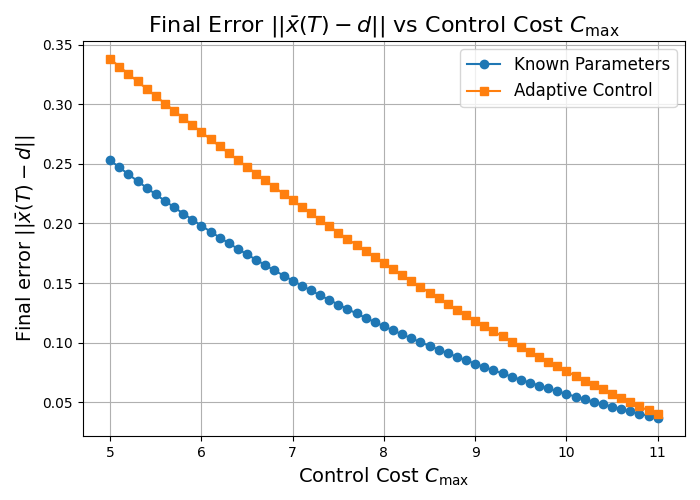}
        \caption{Final error $\epsilon(T)$ vs. control cost $c_u(T)$ under known $\theta$ vs. adaptive control using $\hat{\theta}(t)$.}
        \label{fig:cost_of_learning}
\end{figure}
\begin{figure}[t]
    \centering        \includegraphics[width=0.67\linewidth]{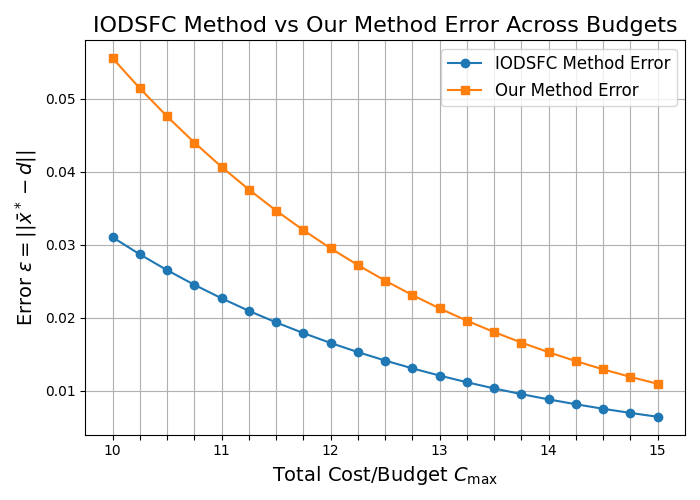}
        \caption{Final error $\epsilon$ across costs comparing the IODSFC baseline and our method.}
        \label{fig:iodsfcvsours}

\end{figure}
\begin{figure*}[t!]
    \centering
    \begin{subfigure}[t]{5.5cm}
        \centering
        \includegraphics[width=5.5cm,height=4cm]{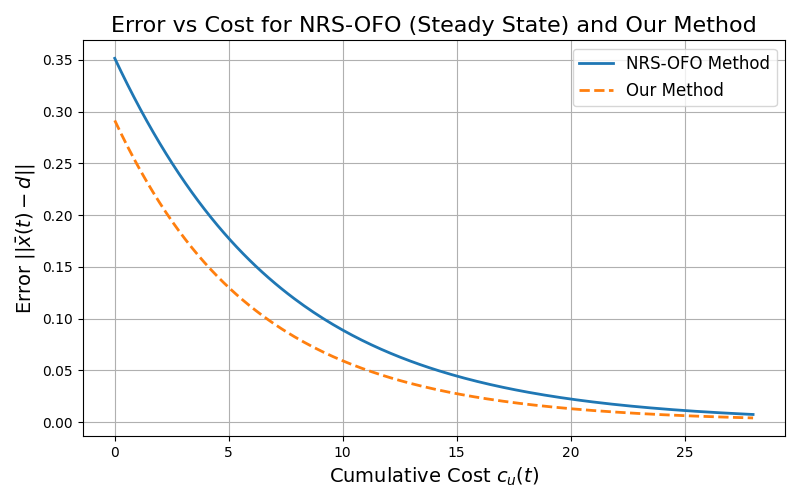}
        
        \caption{Step-wise Gradient descent implementation of NRS-OFO comparison}
        \label{fig:eth-comparison}
    \end{subfigure}\hfill
    \begin{subfigure}[t]{5.5cm}
        \centering
        \includegraphics[width=5.5cm,height=4cm]{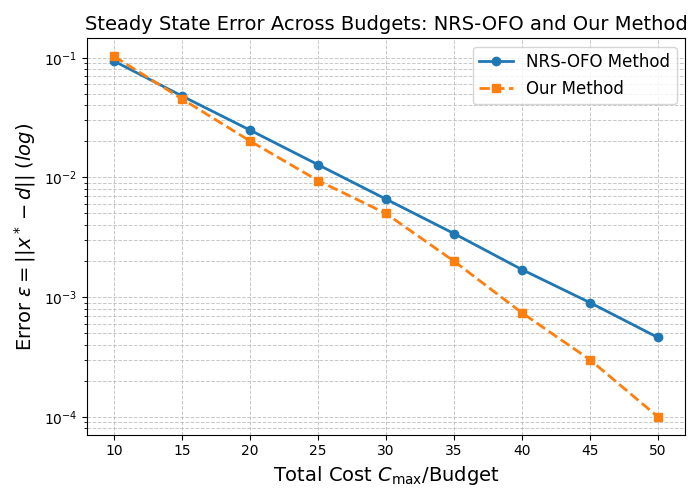}
        \caption{Steady state error across budgets.}
        \label{fig:eth-ss-budget}
    \end{subfigure}\hfill
    \begin{subfigure}[t]{5.5cm}
        \centering
        \includegraphics[width=5.5cm,height=4cm]{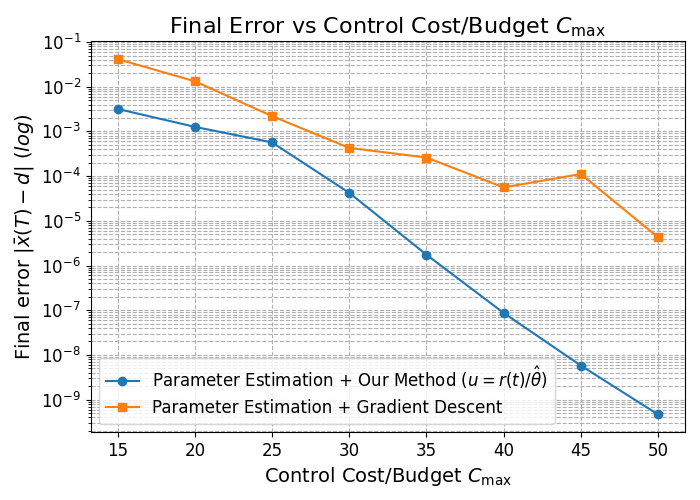}
        \caption{Adaptive Control with Gradient Descent vs Analytical Solving}
        \label{fig:eth-adaptive}
    \end{subfigure}
    \caption{Comparison of our method with the NRS-OFO method under different experimental setups. 
    (a) Step-wise Gradient Descent. (b) Steady State Errors under fixed budgets.
    (c) Parameter estimation of $\theta$ using~\eqref{eq:param-est-algo} with gradient descent versus with our proposed method.}
    \label{fig:eth-results}
\end{figure*}

\textbf{Cost of Parameter Estimation.} We now attempt to quantify the cost of learning the parameter $\theta$. As noted earlier, knowing $\theta$ although not necessary to drive $x(t) \to \mathbf{d}$, improves the rate of convergence. Given the same budget $C_{\max}$ for a fixed time horizon $T$, we can reach a lower $\| x(T) - \mathbf{d}\|$, if $\theta$ is known a priori, as opposed to when it must be estimated online- these results can be seen in Fig. \ref{fig:cost_of_learning}. In the \emph{Known Parameters} case, the control can be set optimally as $u_j(t) = r/\theta_j,\ \forall j$. However in the \emph{Adaptive Control} case the control  is set to $u_j(t) = r/\hat\theta_j(t), \ \forall j$, while $\hat{\theta}$ is updating using parameter identification. 
    The gap between the two curves quantifies the cost of learning: with exact $\theta$  we could achieve a lower error $\|x(T) - \mathbf{d}\|$ for the same control cost. 
    
\textbf{Comparisons with Related Work.} (I) Our setting closely aligns with that of \cite{kth} (Influencing Opinion Dynamics to promote Sustainable Food Choices, IODSFC), 
Under a fixed budget, the IODSFC method computes the \emph{ optimal} control via numerical optimization. In contrast, we derive a closed-form update for $u(t)$, avoiding optimization. For a fair comparison, since they assume known parameters, we assume likewise that the parameter $\theta$ is known for both methods in our simulation. Across varying budgets (Fig.~\ref{fig:iodsfcvsours}), our approach attains performance close to the optimal, with the gap narrowing as the budget increases. Notably, as opposed to IODSFC we also provide analytic convergence conditions (ref. Theorem~\ref{thm:GAS}) which enable this  control design rather than relying on optimization.

(II) The work by \cite{chandrasekaran_network-aware_2024} (Network-aware Recommender System via Online Feedback Optimization, NRS-OFO) proposes a projected-gradient control framework for recommender systems. For a fair comparison we assume in both methods that the state $x(t)$, and the system dynamics are known, and we implement their `Gradient Estimation \& Optimization' procedure (Section 3: Level III).
We adopt the loss function $\mathcal{L}(x(t))=\tfrac{1}{2}\|x(t)-\mathbf{d}\|_2^2$, and minimize it via w.r.t. the control input $U$ which requires the steady-state Jacobian $\nabla_U x^*$.

The NRS-OFO method updates the control input in intervals, and lets the system converge to a steady state in between updates (the control input is constant within each interval). In our model, the only steady state (under bounded $u(t)$) is $x^* = \mathbf{d}$, though in other frameworks the steady state may be slow to reach or misaligned with the target.
To accurately represent this in our simulation of the model, we compute $U(t)$ via gradient descent and keep it constant until convergence or until the cumulative cost of control $c_u(T)$ reaches the budget $C_{\max}$.

As shown in Fig.~\ref{fig:eth-comparison}, for a given budget $C_{\max}=15$ our method converges faster and attains a lower error.
In Fig.~\ref{fig:eth-ss-budget}, we show the steady-state error i.e. $\|x^*-\mathbf{d}\|$ across budgets $C_{\max}\in\{10,15,\dots,50\}$ for both methods. Our method achieves lower error for the same cost because it designs $U(t)$ to satisfy a budget-constrained error guarantee, whereas NRS–OFO does not explicitly incorporate budget.
\noindent Finally, in Fig.~\ref{fig:eth-adaptive} we present an ablation between two online controllers that both use parameter estimation (as in section 4) but differ in their controllers: (i) gradient-descent updates to $U$, and (ii)  our analytic control $u_j=r(t)/\hat{\theta_j}, \ \forall j$. Both estimate $\hat{\theta}$ online while steering opinions toward $\mathbf{d}$, but our method converges faster and achieves a lower steady-state error.
Therefore, our approach outperforms NRS–OFO on achieving lower error for a given cost, while NRS–OFO’s advantage is that it can operate without explicit knowledge of the system dynamics, which our method requires.
\section{CONCLUSIONS}
We presented an online control framework for opinion dynamics with unknown influence parameters, establishing theoretical guarantees on convergence of both the opinion trajectory and parameter estimates. Our analysis gives explicit contraction conditions for target convergence.  In benchmarks against recent methods, our work performs near-optimally while providing formal guarantees for the analytic control design, and does not incur the computational overhead of optimization. Moreover, by focusing on a specific opinion-dynamics model, we achieve faster convergence to the target and direct estimation of the underlying parameters designed for a specified budget, rather than treating them implicitly within model-agnostic approaches.
Future work will include extending estimation beyond the susceptibility parameter to jointly infer the interaction graph and its related parameters. We will also generalize the control and analysis to heterogeneous targets i.e. $\mathbf{d} \neq d\,\mathbf{1}$ for any scalar $d$, and designing budget-aware controllers that ensure progress under non-uniform objectives.



\bibliographystyle{IEEEtran}  
\bibliography{references} 
\addtolength{\textheight}{-12cm}   




\end{document}